\documentclass[draftclsnofoot, onecolumn, 12pt]{IEEEtran} 

\usepackage{graphicx,epsfig}
\usepackage[noadjust]{cite}
\usepackage{mcite}
\usepackage{amsfonts,helvet}
\usepackage{fancyhdr}
\usepackage{threeparttable}
\usepackage{epsf,epsfig}
\usepackage{amsthm}
\usepackage{amsmath}
\usepackage{siunitx}
\usepackage{amssymb}
\usepackage{dsfont}
\usepackage{subfigure}
\usepackage{color}
\usepackage{algorithm}
\usepackage{algpseudocode}
\usepackage{algcompatible}
\usepackage{pifont}
\usepackage{enumerate}
\usepackage{cancel}
\usepackage{lipsum}

\newtheorem{corollary}{Corollary}

\newtheorem{proposition}{Proposition}

\def\bf{{\mathbf{f}}}
\def\bg{{\mathbf{g}}}
\def\bh{{\mathbf{h}}}

\def\bn{{\mathbf{n}}}

\def\bs{{\mathbf{s}}}

\def\bx{{\mathbf{x}}}
\def\by{{\mathbf{y}}}

\def\b0{{\mathbf{0}}}

\def\bF{{\mathbf{F}}}

\def\bH{{\mathbf{H}}}
\def\bI{{\mathbf{I}}}

\def\bP{{\mathbf{P}}}

\def\bR{{\mathbf{R}}}

\def\bU{{\mathbf{U}}}

\def\bW{{\mathbf{W}}}

\def \Tr{ {\rm{Tr}}}

\def\Ptx{P_{\rm{tx}}}
\def \bLambda {\mathbf{\Lambda}}

\newcommand{\figref}[1]{{Fig.}~\ref{#1}}
\newcommand{\sref}[1]{{Section}~\ref{#1}}

\begin{document}

\title{Optimal User Loading in Massive MIMO Systems with Regularized Zero Forcing Precoding}

%
%

\author{
Sungwoo~Park,~\IEEEmembership{Student Member,~IEEE,}
~Jeonghun~Park,~\IEEEmembership{Student Member,~IEEE,}
~Ali~Yazdan Panah,~\IEEEmembership{Member,~IEEE,}
and~Robert~W. Heath, Jr.,~\IEEEmembership{Fellow,~IEEE}
\thanks{S. Park and R. W. Heath are with the Wireless Networking and Communication Group (WNCG), Department of Electrical and Computer Engineering, The University of Texas at Austin, TX, 78701 USA. (e-mail: swpark96@utexas.edu). }
\thanks{A. Y. Panah is with The Connectivity Lab., Facebook, 1 Hacker Way, Menlo Park, CA, 94025 USA (e-mail:ayp@fb.com).}
}

\maketitle \setcounter{page}{1}

\begin{abstract}
We consider a downlink multiuser multiple-input multiple output (MIMO) system employing regularized zero-forcing (RZF) precoding. We derive the asymptotic signal-to-leakage-plus-noise ratio (SLNR) as both the number of antennas and the number of users go to infinity at a fixed ratio. Focusing on the symmetric uncorrelated channels, we show that the SLNR is asymptotically equal to signal-to-interference-plus-noise ratio (SINR) which allows us to optimize the user loading for spectral efficiency. The results show that the optimal user loading varies depending on the channel signal-to-noise ratio (SNR) but is equal to one in both the low or high SNR regimes. 
\end{abstract}


\section{Introduction}\label{sec:intro}

\IEEEPARstart{M}{assive} MIMO systems promise  dramatic gains in sum spectral efficiency, by supporting many simultaneous users. Amazingly in \cite{Marzetta2010}, linear conjugate precoding is optimal when the number of antennas $N$ becomes large while the number of users $K$ is fixed, even though the precoding does not account for inter-user interference. If both $N$ and $K$ become large at a fixed ratio $\alpha = {K}/{N}$, interference-reducing precoding techniques such as zero-forcing (ZF) and regularized zero-forcing (RZF) outperform conjugate precoding. Moreover, in this regime, the effective signal-to-interference-plus-noise ratio (SINR) of RZF, which is a random variable due to the fading channel, converges to a deterministic value. In \cite{Hoydis2013, Wagner2012}, the asymptotic deterministic SINR of RZF was derived considering channel estimation error and antenna correlation. 

In this letter, we derive  asymptotic deterministic signal-to-leakage-plus-noise ratio (SLNR) of RZF  for perfect channel state information assuming antenna correlation
as the    RZF  maximizes the SLNR \cite{Sadek2007}.
 The derived SLNR expression has a significantly simpler form than the asymptotic SINR expression in \cite{Hoydis2013, Wagner2012} allowing us to make multiple observations on system optimality. For example, we show that in the large antenna array regime, user scheduling has no effect on the uncorrelated channel case while the scheduling still plays an important role in the correlated channel case. In addition, we find that the SLNR is asymptotically equal to the SINR when $N \to \infty$ in the symmetric uncorrelated channel case. Leveraging these results, we also derive optimal user loading ($\alpha$) strategies that maximize the sum rate in different signal-to-noise ratio (SNR) regimes.

\section{Asymptotic SLNR of Regularized ZF} Consider a single-cell TDD multi-user system where a base station equipped with $N$ antennas communicates with $K$ single-antenna users. We assume that the base station knows perfect channel state information  
The downlink received signal is modeled as:
\begin{equation}\label{sys_mod}
\by = \bH^*\bx+\bn = \bH^*\bF\bP\bs +\bn,
\end{equation}
where $\bH^*$ is the downlink channel matrix, $\bF$ is a RZF precoding matrix, $\bP$ is a power control matrix and $\bn\sim\mathcal{CN} (\mathbf{0},\sigma^2\bI_K)$ is AWGN. 
The downlink channel is modeled as $\bH^*=\begin{bmatrix} \bh_1 & \cdots & \bh_K \end{bmatrix}^*$,
where $\bh_k \in \mathbb{C}^{N \times 1}$ is the channel vector of user $k$. Considering a correlated channel, $\bh_k$ is modeled as $\bh_k=\bR_k^{\frac{1}{2}} \bh_{w,k}$
where $\bR_k$ is $\mathbb{E}[\bh_k\bh_k^*]$, and $\bh_{w,k} \sim \mathcal{CN} (\mathbf{0}, \bI_N)$. 
Using the RZF with a regularization parameter $\beta$, the precoding matrix $\bF$ is represented as
\begin{equation}\label{rzf_mod}
\bF  =  \begin{bmatrix} \bf_1 & \cdots & \bf_K \end{bmatrix} =  \left(\bH \bH^* + \beta \bI_N \right)^{-1}\bH .
\end{equation}
The total transmit power is $\Ptx$, and we assume that the per-user power is controlled to be identical. In this case, $\bP$ in \eqref{sys_mod} can be expressed as $\bP = {\mathrm{diag}} ([ p_1,p_2,\ldots,p_K ])$
where
\begin{equation}\label{equal_power_Tx}
p_k=\frac{\sqrt{P_{tx}}}{\sqrt{K}\|\bf_k\|} =\sqrt{ \frac{P_{tx}}{ K     \bh_k^* \left(\bH \bH^* + \beta \bI_N \right)^{-2}\bh_k  } } .
\end{equation}

The SLNR of user $k$ can be represented as
\begin{equation}\label{SLNR_def}
\begin{split}
\mathrm{SLNR}_k &= \frac{| \bh_k^* \bf_k p_k |^2}{\sum_{i \neq k} | \bh_i^* \bf_k p_k |^2 + \sigma^2} \\
&=\frac{ \bh_k^*\bW_{\beta}^{-1}\bh_k \bh_k^*\bW_{\beta}^{-1}\bh_k}{ \bh_k^* \bW_{\beta}^{-1} \left(\sum_{i \neq k}  \bh_i\bh_i^*+ \frac{K\sigma^2}{\Ptx} \bI_N \right) \bW_{\beta}^{-1}\bh_k   },
\end{split}
\end{equation}
where $\bW_{\beta} =  \left(\bH \bH^* + \beta \bI_N \right)$.
If the regularization parameter $\beta$ is set to be equal to $\frac{K\sigma^2}{\Ptx}$ as in \cite{Peel2005} the SLNR in \eqref{SLNR_def} can be further simplified as
\begin{equation}\label{SLNR_simple}
\begin{split}
\mathrm{SLNR}_k 
&= \frac{\bh_k^*  \left(\bH \bH^* + \frac{K\sigma^2}{\Ptx} \bI_N \right)^{-1}\bh_k}{1 - \bh_k^* \left(\bH \bH^*  + \frac{K\sigma^2}{\Ptx} \bI_N  \right)^{-1}\bh_k} \\
&=\bh_k^*  \left(\sum_{i \neq k}^{K} \bh_i \bh_i^* + \frac{K\sigma^2}{\Ptx} \bI_N \right)^{-1}\bh_k,
\end{split}
\end{equation}
where the second equality comes from the matrix inversion lemma.
Let $\eta = {\sigma^2}/{\Ptx}$ so that the SNR is $1/\eta$, and $\bh_k=\bR_k^{\frac{1}{2}} \bh_{w,k} =\sqrt{N}\bR_k^{\frac{1}{2}} \bg_{k} $, where $\bg_{k} \sim \mathcal{CN} (\mathbf{0}, \frac{1}{N} \bI_N)$.  If we assume that $\bR_k$ has uniformly bounded spectral norm, i.e., $\lim_{N \to \infty} \| \bR_k \|_2 < \infty$ as in \cite{Hoydis2013, Wagner2012},   the SLNR in \eqref{SLNR_simple} converges to 
\begin{equation}\label{SLNR_det_eq}
\begin{split}
&\mathrm{SLNR}_k = N \bg_k^H \bR_k^{\frac{1}{2}}  \left(N \sum_{i\neq k}^K \bR_i^{\frac{1}{2}} \bg_i\bg_i^*  \bR_i^{\frac{1}{2}} + K\eta  \bI_N \right)^{-1} \bR_k^{\frac{1}{2}} \bg_k \\
& \;\; \xrightarrow{a.s.}  \Tr\left(\bR_k^{\frac{1}{2}}  \left(N \sum_{i\neq k}^K \bR_i^{\frac{1}{2}} \bg_i\bg_i^*  \bR_i^{\frac{1}{2}} + K\eta  \bI_N \right)^{-1} \bR_k^{\frac{1}{2}}  \right) \\
& \;\; \xrightarrow{a.s.}  
\Tr\left(\bR_k \left(N \sum_{i=1}^K \bR_i^{\frac{1}{2}} \bg_i\bg_i^*  \bR_i^{\frac{1}{2}} +K\eta \bI_N \right)^{-1}  \right), 
\end{split}
\end{equation}
as $N$ goes to infinity. In \eqref{SLNR_det_eq},
the first convergence comes from the trace lemma \cite{BookCouillet2011}, and the the second convergence comes from the rank-1 perturbation lemma \cite{BookCouillet2011}.

The random variable SLNR in \eqref{SLNR_det_eq} converges to a deterministic SLNR value, by Theorem 1 in \cite{Wagner2012}, as
\begin{equation}\label{SLNR_converge}
\mathrm{SLNR}_k \xrightarrow{a.s.}  \gamma_k,
\end{equation}
where $\gamma_1,...,\gamma_K$ are the unique nonnegative solutions of 
\begin{equation}\label{SLNR_converge_eq_gamma}
\gamma_k =   \Tr \left(  \bR_k  \left(  \sum_{j=1}^K \frac{\bR_j}{1+\gamma_j } + K\eta \bI_N \right)^{-1} \right).
\end{equation}
These fixed-point equations can be solved using numerical methods, 
and in some special cases there exist closed-form solutions as exemplified by the following corollaries. 

\begin{corollary}
For uncorrelated channels, i.e. $\bR_k = \bI_N, \forall k$, the SLNR is asymptotically equal amongst the users:
\begin{equation}\label{SLNR_converge_uncorr}
\gamma_k =  \frac{-\left( \eta -\frac{N}{K} +1 \right) + \sqrt{\left(\eta  -\frac{N}{K} +1 \right)^2 + 4 \eta \frac{N}{K} } }{2 \eta} ,\forall k .
\end{equation}
\label{corol:SLNR_eq_uncorrelated}
\end{corollary}
\begin{proof}
When $\bR_k = \bI_N, \forall k$, $\gamma_k$ in \eqref{SLNR_converge_eq_gamma} is given by
\begin{equation}\label{SLNR_converge_eq_gamma_uncorr}
\begin{split}
\gamma_k =   \Tr \left(    \left( \sum_{j=1}^K \frac{1}{1+\gamma_j } + K\eta   \right)^{-1} \bI_N \right) = \frac{N}{ \sum_{j=1}^K \frac{1}{1+\gamma_j } + K\eta}, \;\; \forall k,
\end{split}
\end{equation}
which implies that all $\gamma_k$'s have the same value. The value, denoted as $\gamma$, is the positive solution of the equation, and simplifies to \eqref{SLNR_converge_uncorr}.
\end{proof}

Regarding the spatially correlated channel case, we first consider the case when all users have the same $\bR$ matrix.

\begin{corollary}
Let $\bR_k = \bR, \forall k$ with $\Tr(\bR) = N$, and denote the eigenvalues of $\bR$ as $ \lambda_1, ..., \lambda_N  $. Then, the SLNR is asymptotically equal amongst the users and is given by the solution to:
\begin{equation}\label{SLNR_converge_eq_gamma_corr_same}
\gamma =   \sum_{n=1}^{N} \frac{1}{\frac{K}{1+\gamma } + \frac{K\eta}{ \lambda_n} }=\gamma_k, \forall k.
\end{equation}
Moreover, the SLNR is upper bounded by \eqref{SLNR_converge_uncorr}.
\label{corol:SLNR_eq_correlated_worst}
\end{corollary}
\begin{proof}
Let $\ \bR=\bU\bLambda \bU^*$ by eigenvalue decomposition. Then,
\begin{equation}\label{SLNR_converge_eq_gamma_corr_diff}
\begin{split}
\gamma_k 
=   \Tr \left(    \bLambda  \left(   \sum_{j=1}^K \frac{\bLambda }{1+\gamma_j } + K\eta \bI_N \right)^{-1} \right) 
= \sum_{n=1}^{N} \frac{1}{ \sum_{j=1}^K \frac{1}{1+\gamma_j } + \frac{K\eta}{\lambda_n} }, \;\; \forall k,
\end{split}
\end{equation}
which implies $\gamma_1 = ... = \gamma_K =\gamma =   \sum_{n=1}^{N} \frac{1}{\frac{K}{1+\gamma } + \frac{K\eta}{\lambda_n} } $. Since $\frac{1}{a+b \lambda^{-1}}$ is a concave function of $\lambda$, $\gamma$ is upper bounded as
\begin{equation}\label{SLNR_converge_eq_gamma_corr_diff_ub}
\begin{split}
\gamma = \sum_{n=1}^{N} \frac{1}{  \frac{K}{1+\gamma } + \frac{K\eta}{\lambda_n} } \leq  \frac{N}{  \frac{K}{1+\gamma } + \frac{NK\eta}{\sum_{n=1}^{N} \lambda_n} } =  \frac{N}{  \frac{K}{1+\gamma } + K\eta },
\end{split}
\end{equation}
where equality holds if $\bR=\bI$. 
\end{proof}
Next, we consider the case of $\frac{1}{K}\sum_{k=1}^{K} \bR_k = \bI$. This scenario occurs when there are a large number of users and users are properly selected via scheduling. For example, in the exponential correlation model, i.e., $[\bR_k]_{m,n}=\rho^{|m-n|} e^{j  (m-n) \theta_k}$ \cite{Loyka2001}, if users are selected such that $\theta_k = 2\pi k / K$, then $\frac{1}{K}\sum_{k=1}^{K} \bR_k = \bI$ because $\sum_{k=1}^{K} \rho^{|q|} e^{j  \frac{2 \pi q k}{K}}$ is $0$ for $q\neq 0$ and $K$ for $q=0$.
\begin{corollary}
Let $\frac{1}{K}\sum_{k=1}^{K} \bR_k = \bI$ with $\Tr(\bR_k)=N, \forall k$. Then, the SLNR is asymptotically equal amongst the users and is equal to that of the uncorrelated case in \eqref{SLNR_converge_uncorr}.
\label{corol:SLNR_eq_correlated_best}
\end{corollary}
\begin{proof}
Suppose that $\gamma_1 = \cdots = \gamma_K = \gamma$. Using $\sum_{k=1}^{K} \bR_k = K\bI$, the fixed-point equations in \eqref{SLNR_converge_eq_gamma} can be written as
\begin{equation}\label{SLNR_converge_eq_gamma_uncorr}
\begin{split}
\gamma &= \frac{\Tr(\bR_k)}{\frac{K}{1+\gamma} + K\eta}= \frac{N}{ \frac{K}{1+\gamma} + K\eta}, \;\; \forall k,
\end{split}
\end{equation}
and the solution is given by \eqref{SLNR_converge_uncorr}. Since \eqref{SLNR_converge_eq_gamma} has a unique solution, the trivial solution of  $\gamma_1 = \cdots = \gamma_K$ can be regarded as that unique solution. 
\end{proof}

\textit{Observations:} Corollary 3 and 4 indicate that the SLNR of the correlated case is generally worse than the uncorrelated case, but can asymptotically approach the SLNR of the uncorrelated case if users are well selected (via scheduling). As such, the SLNR of the uncorrelated channels does not depend on scheduling due to channel hardening in the massive MIMO regime. Uer scheduling, however, could have an impact in the correlated channel case. This phenomenon may not hold for the few-antenna regime since the SLNR values are generally not deterministic in this case. 

\textit{Remark:} In a single-path channel where $\mathrm{rank}(\bR_k)=1$, the interference does not exist when users are selected such that user channels are orthogonal to each other. 
In this extreme case, the multiuser MIMO system can be regarded as parallel single user MISO systems without any interference. 
In addition, the SLNR does not converge to a deterministic value and still remains a random variable dependent on short-term channel fading due to the lack of diversity. In this letter, we focus on general multi-path channel cases in which the inter-user interference cannot be perfectly eliminated by only scheduling itself.

\section{Optimal User Loading in Regularized ZF}\label{sec:opt_user_loading}

In this section, we will analyze the optimal user loading $\alpha$ that maximizes the sum rate in the symmetric uncorrelated case as in the Corollary 1. We begin by showing how the SINR converges to our derived SLNR as $N \to \infty$. 

The SINR is given by
\begin{equation}\label{SINR_def}
\mathrm{SINR}_k = \frac{| \bh_k^* \bf_k p_k |^2}{\sum_{i \neq k} | \bh_k^* \bf_i p_i |^2 + \sigma^2},
\end{equation}
where the only difference from SLNR of \eqref{SLNR_def} is the interference term in the denominator. The $i$-th interference term of the SINR is asymptotically equal to that of the SLNR as
\begin{equation}\label{interference_term}
 |p_i |^2| \bh_k^* \bf_i |^2 =  |p_i |^2| \bh_i^* \bf_k |^2 \xrightarrow{a.s.}   |p_k|^2| \bh_i^* \bf_k |^2,
\end{equation}
where the first equality is from: $| \bh_k^* \bf_i|^2 = | \bh_i^* \bf_k |^2=|\bh_k^* \left(\bH \bH^* + K\eta \bI_N \right)^{-1} \bh_i|^2$, and the second (almost surely) property comes from the fact that $|p_k|^2 $ is asymptotically equal for all users:
\begin{equation}\label{pi_converge}
\begin{split}
|p_k|^2 &= \frac{P_{tx}}{ K     \bh_k^* \left(\bH \bH^* + K\eta \bI_N \right)^{-2}\bh_k  } \xrightarrow{a.s.} c,
\end{split}
\end{equation}
where 
\begin{equation}\label{pi_converge_c}
c=  \frac{N P_{tx}\left(1 +   \Tr \left( \left( \sum_{i =1}^K \bg_i\bg_i^*   + \frac{K\eta}{N}  \bI_N \right)^{-1} \right) \right)^2}{ K     \Tr \left( \left( \sum_{i =1}^K \bg_i\bg_i^*   + \frac{K\eta}{N}  \bI_N \right)^{-2} \right) }.
\end{equation}
The result follows by applying the matrix inversion lemma, the trace lemma and the rank-1 perturbation lemma.

Using this result, the optimal user loading that maximizes the sum rate can be formulated as
\begin{equation}\label{K_optimization_problem}
x^{\star}=\arg\max_{x \geq 1} \frac{\log \left(1+\gamma(x,\eta) \right)}{x} ,
\end{equation}
where $x$ is defined as $x=\frac{1}{\alpha} = \frac{N}{K}$ for notational convenience and $\gamma(x,\eta) $ is given by \eqref{SLNR_converge_uncorr}.
Let the objective function in \eqref{K_optimization_problem} be $f(x,\eta)$. The first derivative is given by
\begin{equation}\label{K_optimization_problem_dfx}
\begin{split}
\frac{\partial f(x,\eta)}{\partial x} =\frac{1}{x\sqrt{\left( x+ \eta  - 1 \right)^2 + 4 \eta} } 
-\frac{1}{x^2}  \log \left(\frac{  x + \eta - 1  + \sqrt{\left( x+ \eta  - 1 \right)^2 + 4 \eta} }{2 \eta}  \right).
\end{split}
\end{equation}
It can be shown that the derivative $\frac{\partial f(x,\eta)}{\partial x} $ has the following properties for any $\eta$:  $\lim_{x \to -\infty}\frac{\partial f(x,\eta)}{\partial x}   =\lim_{x \to \infty} \frac{\partial f(x,\eta)}{\partial x}   =0 $, and $\frac{\partial f(x,\eta)}{\partial x} =0$ has only one solution in $x \in (-\infty, +\infty)$, and $\frac{\partial f(x,\eta)}{\partial x} >0$ if $x<x^{\star}$ and $\frac{\partial f(x,\eta)}{\partial x} <0$ if $x>x^{\star}$ where $x^{\star} $ is the solution of $\frac{\partial f(x,\eta)}{\partial x}  =0$. Therefore, $f(x)$ has a global maximum at $x^{\star}$ satisfying $\frac{\partial f(x,\eta)}{\partial x}  =0$. 

While a simple closed-form solution to $\frac{\partial f(x,\eta)}{\partial x}  =0$ may not exist, we can gain some insight on the optimal solution owing to the following propositions.

\begin{proposition}
$x^{\star}=1$ for $\eta \geq \eta_{o}$, i.e. low SNR regime, where $\eta_{o}=0.3256$.
\end{proposition}
\begin{proof}
Given the constraint of $x\geq 1$, the maximum occurs at $x=1$ if $x^{\star} \leq 1$ since $f(x,\eta)$ decreases monotonically at $x \geq 1$. Using the aforementioned properties, the condition of $x^{\star} \leq 1$ is equivalent to ${\partial f(x,\eta)}/{\partial x}|_{x=1} \leq 0$. Setting $x=1$ in \eqref{K_optimization_problem_dfx}, we seek $\eta_{o}$ as the solution to $$\sqrt{\eta^2 + 4 \eta}\log \left(\frac{  \eta  + \sqrt{ \eta ^2 + 4 \eta} }{2 \eta}  \right)-1=0,$$ which numerically solves to $\eta_{o}=0.3256$.
\end{proof}
\textit{Observation:} Proposition 1 implies that the optimal user loading $\alpha \left(={1}/{x}\right)$ is equal to $1$ if the SNR is less than $10\log_{10}(1/0.3256) = 4.78$ dB. 
\begin{proposition}
For $0 \ll \eta < \eta_{o}$, $x^{\star}$ can be approximated as
\begin{equation}\label{low_SNR_approx}
x^{\star} \approx c_\eta+\sqrt{c_\eta^2-(1-2c_\eta)(\eta+3)},
\end{equation}
where $c_\eta=1-\frac{1}{2}\sqrt{  \eta  ^2 + 4 \eta}\log \left(\frac{   \eta   + \sqrt{ \eta ^2 + 4 \eta} }{2 \eta}  \right)$.
\end{proposition}
\begin{proof}
At SNR of $4.78$ dB, i.e. $\eta \approx \eta_{o}$, we know that $x^{\star} \approx 1$ and ${\partial f(x,\eta)}/{\partial x} |_{x=x^{\star}} = 0$, thus \eqref{low_SNR_approx} follows from applying the first-order Taylor approximation to the denominator of the first term and the second term in \eqref{K_optimization_problem_dfx}, at $x=1$. 
\end{proof}

\begin{proposition}
For $\eta \approx 0$, i.e. high SNR region, $x^{\star}$ can be approximated as
\begin{equation}\label{K_optimization_problem_dfx_approx_highSNR_sol}
\begin{split}
x^{\star}= 1- \eta + \eta e^{1+\mathcal{W}\left({(1-\eta)}/{\eta e} \right)},
\end{split}
\end{equation}
where $\mathcal{W}(\cdot)$ is the Lambert W-function defined as $z=\mathcal{W}(z)e^{\mathcal{W}(z)}$, and $x^{\star}$ is an increasing function of $\eta$.
\end{proposition}
\begin{proof}
At $\eta \approx 0$, $ {\partial f(x,\eta)}/{\partial x}$ can be approximated as
\begin{equation}\label{K_optimization_problem_dfx_approx_highSNR}
\begin{split}
 \frac{\partial f(x,\eta)}{\partial x} &=\frac{1}{x^2}\left( \frac{x}{ x+ \eta  - 1  } -\log \left(\frac{  x + \eta - 1   }{\eta}  \right) \right),
\end{split}
\end{equation}
and the solution of $ {\partial f(x,\eta)}/{\partial x}=0$ is given by \eqref{K_optimization_problem_dfx_approx_highSNR_sol}.
\end{proof}
It can be shown that $x^{\star}$ in \eqref{low_SNR_approx} is a decreasing function of $\eta$ while  $x^{\star}$ in \eqref{K_optimization_problem_dfx_approx_highSNR_sol} is an increasing function of $\eta$. The following proposition introduces another useful property of the optimum. 
\begin{proposition}
$x^{\star}<3(2\sqrt{3}-3)$ for any $\eta$.
\end{proposition}
\begin{proof}
It can be shown that ${\partial^2 f(x,\eta)}/{\partial x \partial \eta}<0$ and $\lim_{\eta \to \infty} {\partial f(x,\eta)}/{\partial x} =0$ if $x>3(2\sqrt{3}-3)$. Therefore, ${\partial f(x,\eta)}/{\partial x}<0$, which concludes the proof. 
\end{proof}
\textit{Observation:} Proposition 4 provides a loose upper bound for $x^{\star}$. A tight upper bound, $x^{\star}_{\mathrm{UB}}$, is given by the solution to ${\partial f(x,\eta)}/{\partial x}|_{\eta = \eta^{\star}(x)} =0$, where $\eta^{\star}(x)$ is the smallest solution to ${\partial^2 f(x,\eta)}/{\partial x \partial \eta } =0$ with regard to $\eta$ for $1<x< 3(2\sqrt{3} -3)$. The numerical solution to this problem is $x^{\star}_{\mathrm{UB}}=1.3315$, which implies that a tight lower bound on the optimal user loading $\alpha$ is $0.75$.

To summarize the properties of optimal user loading ${K}/{N}$
\begin{enumerate} 
\item{In the low SNR regime, the optimal loading is a decreasing function of SNR and equals $1$ when SNR $<4.78$ dB. }
\item{In the high SNR regime, the optimal loading is an increasing function of SNR and approaches a value of $1$ as $\text{SNR} \to \infty$. }
\item{There exists a lower bound on the optimal user loading equal to roughly $0.75$.}
\end{enumerate}

\section{Numerical Results}\label{sec:numerical_results}

In this section, we support our results with simulations. \figref{fig:Fig_SLNR_CDF} plots the asymptotic SLNR \eqref{SLNR_converge_eq_gamma} along with the CDF curves for the instantaneous SLNR and SINR in the symmetric uncorrelated channel case. The SNR is set to $20$ dB for all cases. Note how the random variable SLNR approaches the deterministic value of the derived SLNR as $N$ becomes large. In addition, the SINR converges to the asymptotic SLNR value as $N$ is large, and the convergence rate is faster when the user loading ${K}/{N}$ is small.

\begin{figure}[!t]
	\centerline{\resizebox{0.60\columnwidth}{!}{\includegraphics{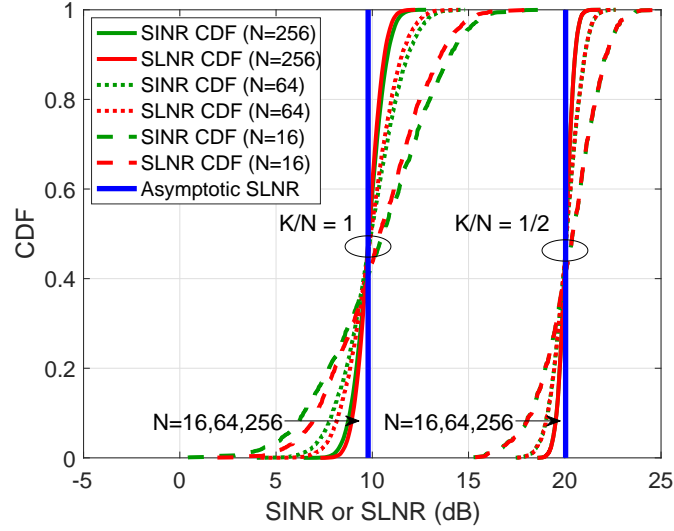}}}
	\caption{Asymptotic SLNR and CDF of SLNR and SINR per user}
	\label{fig:Fig_SLNR_CDF}
\end{figure}

\begin{figure}[!t]
	\centerline{\resizebox{0.60\columnwidth}{!}{\includegraphics{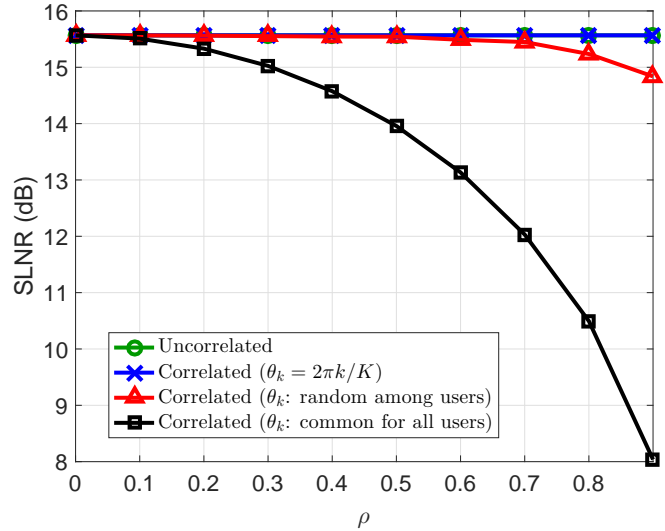}}}
	\caption{Correlation coefficient $\rho$ versus asymptotic SLNR}
	\label{fig:Fig_SLNR_correlated}
\end{figure}

\figref{fig:Fig_SLNR_correlated} compares the correlated and uncorrelated channel cases for $N=128$, $\alpha = 3/4$, and SNR$=20$ dB. We adopt an exponential correlation model \cite{Loyka2001} as $[\bR_k]_{m,n}=\rho^{|m-n|} e^{j \theta_k}$ for user $k$,  assuming a uniform linear array. The correlation factor between adjacent antennas, $\rho$, is assumed to be common for all users. The $\theta_k$, which is associated with the angle of departure of user $k$, is differently set according to three different scenarios; evenly distributed in $[0,2\pi]$ such as $\theta_k=2\pi k/K$, uniformly randomly distributed in $[0, 2\pi]$, or fixed for all users at $\theta$. 
Note that when $\theta_k$ is randomly distributed users have disparate SLNRs so we plot the \textit{average} SLNR. 
The figure shows that the randomly distributed case is much closer to the evenly distributed case rather than to the common $\theta$ case. 
This indicates that the case of $\bR_k=\bR, \forall k$ in Corollary 2, which is sometimes used in some papers, is not a reasonable correlated channel scenario in multiuser MIMO systems. 

The optimal user loading $\alpha\left(={1}/{x}={K}/{N}\right)$ is simulated in \figref{fig:Fig_optimal_user_loading} versus SNR. The exact value is obtained by brute force numerical methods. As shown in \sref{sec:opt_user_loading}, $\alpha=1$ when SNR is less than $4.78$ dB. The high approximation expression in \eqref{K_optimization_problem_dfx_approx_highSNR_sol} is the same as the optimal loading expression for the ZF case in \cite{Wagner2012} which is expected since RZF converges to ZF at high SNR. Finally, it is worth noting that although $\alpha \to 1$ when SNR$\to \infty$, the convergence rate is rather slow.

\begin{figure}[!t]
	\centerline{\resizebox{0.60\columnwidth}{!}{\includegraphics{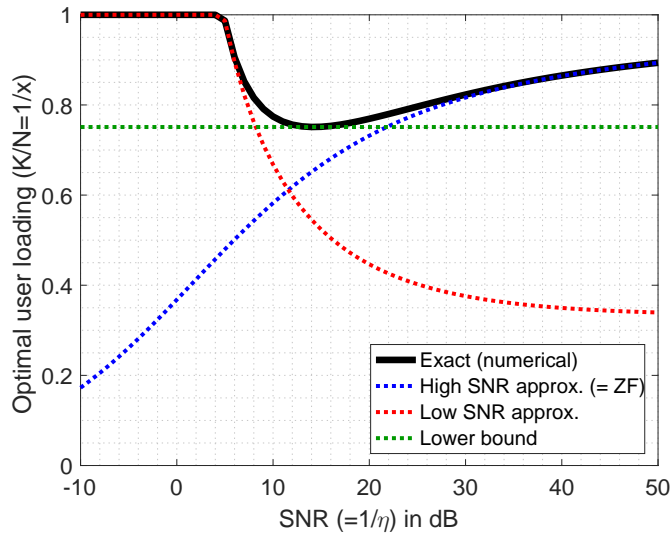}}}
	\caption{Optimal user loading $K/N$ versus SNR}
	\label{fig:Fig_optimal_user_loading}
\end{figure}

\section{Conclusion}\label{sec:conclusions}

In this letter, we derived a simple expression for the asymptotic SLNR under correlated channels when RZF is applied in large antenna systems. We showed that the performance under correlated channels can approach the uncorrelated channel case and we also derived the optimal user loading maximizing the sum rate in the symmetric uncorrelated case using asymptotic equivalence of SINR and SLNR.

\bibliographystyle{IEEEtran}
\bibliography{Refbib_SLNR}

\end{document}